\documentclass[11pt]{article}

\usepackage{mathptmx}

\usepackage{amsmath,amssymb,amsfonts,amsthm} 
\usepackage{mathrsfs,latexsym} 

\usepackage{calrsfs}
\DeclareMathAlphabet{\pazocal}{OMS}{zplm}{m}{n}

\usepackage{color}
\usepackage{cite}
\usepackage{epsfig}
\usepackage{graphicx}

\usepackage{bm}

\usepackage{cite}

\newtheorem{theorem}{Theorem}[section]

\newtheorem{lemma}[theorem]{Lemma}
\newtheorem{proposition}[theorem]{Proposition}

\numberwithin{equation}{section}

\newcommand{\CC}{{\mathbb C}}
\newcommand{\RR}{{\mathbb R}}

\newcommand{\NN}{{\mathbb N}}

\newcommand{\ZZ}{{\mathbb Z}}

\newcommand{\Cc}{{\mathcal{C}}}
\newcommand{\Dc}{{\mathcal{D}}}

\newcommand{\bip}{{\bm p}}

\newcommand{\fV}{{\mathfrak V}}

\newcommand{\Times}{\mbox{\Large $\times$}}      
\newcommand\meddot{\scalebox{0.4}{\textbullet}}  
\newcommand{\DOT}[1]{\overset{\meddot}{#1}}      

\newcommand{\Dom}{\pazocal{D}}                   
\newcommand{\Hil}{\pazocal{H}}                   
\newcommand{\Surf}{\pazocal{S}}                  
\newcommand{\Loop}{\pazocal{L}}                  
\newcommand{\Region}{\pazocal{O}}                
\newcommand{\Poin}{\pazocal{P}}

\def\eg{{\it e.g.\ }}
\def\ie{{\it i.e.\ }}
\def\viz{{\it viz.\ }}

\begin{document} 

%

\title{
The universal C*-algebra of the electromagnetic field \  
II. Topological charges and spacelike linear fields\\[3mm]
{\large \it Dedicated to Karl-Henning Rehren on the occasion of his 60th 
birthday} \\[4mm]}
\author{Detlev Buchholz${}^{(1)}$,
\ Fabio Ciolli${}^{(2)}$, \ Giuseppe Ruzzi${}^{(2)}$ \
and \ Ezio Vasselli${}^{(2)}$ \\[20pt]
\small 
${}^{(1)}$ Institut f\"ur Theoretische Physik, Universit\"at G\"ottingen, \\
\small Friedrich-Hund-Platz 1, 37077 G\"ottingen, Germany\\[5pt]
\small
${}^{(2)}$ 
Dipartimento di Matematica, Universit\'a di Roma ``Tor Vergata'' \\
\small Via della Ricerca Scientifica 1, 00133 Roma, Italy \\
}
\date{}

\maketitle

{\small 
\noindent {\bf Abstract.} 
Conditions for the appearance of topological charges
are studied in the framework of the 
universal C*-algebra of the electromagnetic field,
which is represented in any theory describing electromagnetism. It 
is shown that non-trivial topological charges, 
described by pairs of 
fields localised in certain topologically  
non-trivial spacelike separated regions, can appear in regular 
representations of the algebra only if the fields depend non-linearly 
on the mollifying test functions. On the other hand, 
examples of regular vacuum representations 
with non-trivial topological charges are 
constructed, where the underlying field still 
satisfies a weakened form of ``spacelike linearity''.  
Such representations also appear in the presence of 
electric currents. The status of topological charges 
in theories with several types of electromagnetic fields, 
which appear in the short distance (scaling) 
limit of asymptotically free non-abelian 
gauge theories, is also briefly discussed.
 \\[1mm]
 \\[1mm]
}

\section{Introduction}

A universal C*-algebra 
for the description of the electromagnetic field has 
recently been constructed in \cite{BuCiRuVa}.
It was argued there that representations of this
algebra appear in any theory describing electromagnetism. 
The algebra does not contain any specific dynamical information, yet  
such information can be obtained from it by (i) choosing in  
its dual space some suitable pure state, (ii) proceeding to its GNS 
representation and (iii) considering the quotient of the algebra with 
regard to the kernel of this representation. In this way one obtains 
all theories of the electromagnetic field which comply with the Haag-Kastler
axioms \cite{HaKa}. 

It was observed in this analysis that there exist representations of
the universal algebra with non-trivial topological charges. These charges 
are given by the commutator of operators, describing the 
intrinsic (gauge invariant) vector potential, which have  
their supports in 
certain spacelike separated, topologically non-trivial regions
(precise definitions are given below). As was shown in \cite{BuCiRuVa}, 
such commutators need not vanish, but they are elements of the centre 
of the algebra. The arguments given in \cite{BuCiRuVa} for the 
existence of states where these commutators have values different 
from zero, indicating non-trivial topological charges, led   
only to an abstract existence theorem, however. In particular, 
the conditions for the existence of these charges 
remained unclear. It is the aim 
of the present investigation to clarify this point. 

We will show that in all \textit{regular} pure states on the 
algebra, where one can define the vector potential, the topological charges 
vanish whenever this potential depends linearly on the underlying test 
functions. In particular, such charges cannot appear in the Wightman
framework. Yet the condition of unrestrained linearity of quantum fields on 
test functions does not have a 
clear-cut operational basis and seems more a matter of convenience.
In particular, the Haag-Kastler framework of quantum field theory
does not rely on such a condition and non-linear quantum fields  
already appeared in other contexts, cf. for example 
\cite[Eq.\ 4]{BuMaPaTo} and \cite[Rem.\ 6.1]{BrDuFr}. 

Within the present framework of the universal algebra, 
we will exhibit regular vacuum states with non-trivial 
topological charges.
There the resulting vector potentials are homogeneous on the space of 
test functions, but they have the property of additivity 
only for test functions with spacelike separated supports 
(\ie the potentials are \textit{spacelike linear}).
Our first example is based on a reinterpretation of the theory of the
free electromagnetic field. We then show that such states 
also exist in the presence of non-trivial electric currents. 

We also consider theories of several 
electromagnetic fields which are described by suitable extensions 
of the universal algebra. One may expect that these
theories cover the scaling (short distance) limit of non-abelian 
gauge fields in view of their expected property of asymptotic freedom. 
It is of interest that in these theories there exist 
regular vacuum states with non-trivial topological charges 
where the corresponding vector 
potentials depend linearly on test functions. The topological charges
are given there by the commutator of spacelike separated
operators describing different potentials. It is an intriguing question 
whether such charges might manifest themselves already at finite 
scales in gauge theory. 

In the subsequent section we recall from \cite{BuCiRuVa} 
some basic properties of the universal algebra. 
Sec.~3 contains the proof that regular pure states carrying 
a non-trivial topological charge give rise to non-linear 
vector potentials. In Sec.~4 we exhibit a vacuum state 
with vanishing electric current which carries a topological charge. 
That such vacuum states also exist for non-trivial electric 
currents is shown in Sec.~5. Examples of states carrying  
a topological charge in theories with several fields are presented 
in Sec.~6. The article closes with some brief conclusions.

\section{Preliminaries}

Conventionally, the electromagnetic field $F$ is described 
by an operator valued map $f \mapsto F(f)$, where 
$f \in \Dc_2(\RR^4)$ are compactly supported real test functions
with values in the antisymmetric tensors of rank two 
\cite{Steinmann, Strocchi2}.
As is well known, the homogeneous Maxwell equations imply that 
this field can conveniently be described by an intrinsic 
(gauge invariant) vector potential $g \mapsto A(g)$, where 
$g \in \Cc_1(\RR^4)$, the space of real vector valued 
test functions that are co-closed, \ie satisfy $\delta g = 0$. 
Here $\delta$ denotes the co-derivative (generalised divergence) 
which is related to the exterior derivative $d$ (generalised curl)
by $\delta = - \star d \star$ and $\star$ is the Hodge operator. 
Recalling that $\delta^2 = 0$, 
the electromagnetic field and the 
potential are related  by $F(f) \doteq A(\delta f)$,
$f \in \Dc_2(\RR^4)$, and the electric current is defined by 
$j(h) \doteq A(\delta d h)$ for real vector valued test 
functions $h \in \Dc_1(\RR^4)$. 

In \cite{BuCiRuVa} the properties of the intrinsic vector potential 
were cast into a C*-algebraic setting by formally proceeding 
from $A(g)$ to unitary operators $V(a,g) \, \hat{=} \, e^{\, iaA(g)}$.
More precisely, one proceeds from the free *-algebra, generated
by the symbols $V(a,g)$, where $a \in \RR$, $g \in \Cc_1(\RR^4)$,  
and takes its quotient with regard to the ideal generated by the relations
\begin{eqnarray}
& \label{a1} V(a_1,g) V(a_2,g) = V(a_1+a_2,g) \, , 
\ \ V(a,g)^* = V(-a,g) \, , \ \ V(0,g) = 1 & \\
& \label{a2} 
V(a_1 , g_1)V(a_2 , g_2) = 
V(1, a_1 g_1 + a_2 g_2) \, \ \
\mbox{if} \ \ \mbox{supp} \, g_1 \, \Times \, \mbox{supp} \, g_2 &  
\\
& \label{a3} 
\lfloor V(a,g) \, , \lfloor V(a_1,g_1), V(a_2,g_2) \rfloor \rfloor = 1  \, \ \
\mbox{for any $g$ if}
\ \ \mbox{supp} \, g_1 \perp \mbox{supp} \, g_2 \, . &
\end{eqnarray}
Relation \eqref{a1}
subsumes the algebraic properties of unitary one-parameter groups 
$a \mapsto V(a,g)$, expressing the idea that one is dealing with the
exponential functions of the potential, mollified with test functions 
$g \in \Cc_1(\RR^4)$. Relation~\eqref{a2} encodes restricted 
linearity and locality properties of the vector potential, where the symbol 
$\, \Times$ marks pairs of regions that can be separated by two opposite 
characteristic wedges 
(\eg spacelike separated double cones).
The symbol $\lfloor \cdot , \cdot \rfloor$ 
in relation~\eqref{a3} indicates the group 
theoretic commutator of unitary operators;
this relation embodies the information that 
the commutator of operators which are localised in arbitrary 
spacelike separated regions, marked by the symbol  
$\perp$, is a central element. These operators therefore
determine in general superselected quantities which, in view of their 
topological nature, are called topological charges.
Note that these conditions are slightly weaker 
than the corresponding ones in \cite{BuCiRuVa}. 
They simplify the discussion of the topological charges which 
are of interest here. 

As has been shown in  \cite{BuCiRuVa}, the *-algebra $\fV_0$ generated by 
the unitaries $V(a,g)$ with $g \in \Cc_1(\RR^4)$, $a \in \RR$, can be 
equipped with a C*-norm generated by all of its 
GNS representations. Proceeding to the completion of $\fV_0$ with regard
to this norm, one obtains the universal C*-algebra $\fV$ of the 
electromagnetic field. 

As already mentioned, the algebra $\fV$ does not 
contain any dynamical information. But since the Poincar\'e 
transformations $P \in \Poin_+^\uparrow$ act on $\fV$ by 
automorphisms~$\alpha_P$ which are defined on the generating unitaries
by $\alpha_P(V(a,g)) = V(a,g_P)$, where $g_P$ is the Poincar\'e transformed
test function $g$, one can identify 
the vacuum states in the dual space of $\fV$. Picking any 
such (pure) vacuum state $\omega$ 
and proceeding to its GNS representation $(\pi, \Hil, \Omega)$,
one obtains a faithful representation of the quotient algebra 
$\fV / \mbox{ker} \pi$, where the Poincar\'e transformations are 
unitarily implemented. In this way one can in principle describe 
any dynamics of the electromagnetic field in a manner 
which is compatible with the Haag-Kastler axioms \cite{BuCiRuVa}. 

If one wants to recover in a representation 
from the unitaries $V(a,g)$ the underlying vector potential,  
one has to restrict attention to states $\omega$ that
are (strongly) regular. This means that the functions 
$a_1, \dots , a_n \mapsto \omega(V(a_1, g_1) \cdots V(a_n, g_n))$
are smooth for arbitrary test functions $g_1, \dots , g_n \in 
\Cc_1(\RR^4)$ and $n \in \NN$. One then finds \cite{BuCiRuVa} that in 
the corresponding GNS representation $(\pi, \Hil, \Omega)$
one has $\pi(V(a,g)) = e^{ia A_\pi(g)}$ for $a \in \RR$, $g \in \Cc_1(\RR^4)$.
Here $A_\pi(g)$ are selfadjoint operators with a common core 
$\Dom \subset \Hil$ that is stable under their action
and includes the vector $\Omega$. Moreover, as a consequence of relation 
\eqref{a2}, these operators are \textit{spacelike linear} 
in the sense that they satisfy on $\Dom$ the equality 
$a_1 \, A_\pi(g_1) + a_2 \, A_\pi(g_2) = A_\pi(a_1 g_1 + a_2 g_2)$
whenever $\mbox{supp} \, g_1 \Times \mbox{supp} \, g_2$. 

If the potential resulting
from some regular pure state is also linear in the usual sense, then 
this state carries no topological charges, as is shown in 
the subsequent section. More precisely, one
then has $[A_\pi(g_1), A_\pi(g_2)] = 0$ whenever $g_1, g_2$ have 
spacelike separated supports in linked, loop-shaped regions. Yet, as 
we shall see, there exist also regular vacuum states where these 
commutators have non-trivial values and the underlying fields 
are still spacelike linear. 

\section{Symplectic forms and absence of topological charges
\label{no-go}}

In order to clarify the conditions for the existence of non-trivial
topological charges, we consider in this section regular 
pure states $\omega$ on the universal algebra $\fV$ with corresponding
GNS representation $(\pi, \Hil, \Omega)$; note that we do not require
here that these states describe the vacuum. In complete analogy to the
preceding discussion, the regularity of a state $\omega$ implies that 
for any $g \in \Cc_1(\RR^4)$ there exists a selfadjoint operator $A_\pi(g)$ 
in the corresponding representation. These 
operators  are the generators of the 
unitary groups $a \mapsto \pi(V(a,g))$ and have a common 
stable core $\Dom$ that includes the vector $\Omega$. It then follows
from relation \eqref{a3} that whenever 
$\mbox{supp} \, g_1 \perp \mbox{supp} \, g_2$, the commutator
$[A_\pi(g_1), A_\pi(g_2)]$ is affiliated with the centre of 
the weak closure $\pi(\fV)^-$ of the represented algebra. Since the
underlying state was assumed to be pure, the elements of the centre, 
hence also these commutators, are  multiples of the identity. 
So their expectation values do not depend on the chosen state 
within the representation and this brings us to define the real form 
\begin{equation} \label{symplectic}
\sigma_\pi(g_1, g_2) \doteq i \, \langle \Omega, [A_\pi(g_1), A_\pi(g_2)] \, 
\Omega \rangle \quad \text{for any} \quad g_1, g_2 \in \Cc_1(\RR^4) \, .
\end{equation}
Now according to relation \eqref{a2}, the operators $A_\pi$ are in general  
only spacelike linear. But, depending on the choice of state, they may 
also be (real) linear on $\Cc_1(\RR^4)$. In the latter case the 
form $\sigma_\pi$
given above defines a bilinear, skew symmetric (symplectic) form
on $\Cc_1(\RR^4)$. 

We shall show that such symplectic forms 
$\sigma_\pi$ vanish for any pair of test functions 
{$g_1, g_2 \in \Cc_1(\RR^4)$} having 
their supports in certain spacelike separated, linked loop-shaped regions. 
Whence the corresponding topological charges vanish. An open, bounded  
region $\Loop \subset \RR^4$ is said to be 
loop-shaped if there exists in its interior some spacelike (hence simple)
loop 
$[0,1] \ni t \mapsto \gamma(t)$, consisting of points which are spacelike
separated from each other, to which it can continuously be retracted.
Thus, $\gamma$ is a deformation retract of~$\Loop$ and therefore
homotopy equivalent (homotopic) to~$\Loop$.  
Simple examples of linked loop-shaped regions are 
depicted in figure 1. 

\begin{figure}[h] 
\centering 
\epsfig{file=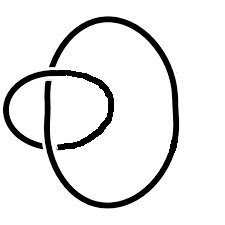,height=30mm}
\epsfig{file=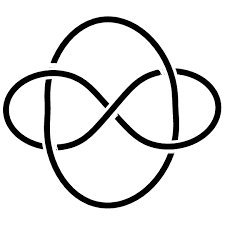,height=30mm}
\caption{Linked loop-shaped regions: Hopf link and Whitehead link} 
\label{fig1} 
\end{figure}

For the proof of these statements, let $\Loop \subset \RR^4$
be loop-shaped with corresponding loop $\gamma$ and let 
$\Region_0 \subset \RR^4$ be a sufficiently small neighbourhood 
of the origin such that 
\mbox{$(\Region_0 + \gamma) \subset \Loop$}. Then, for any real scalar 
test function $s \in \Dc_0(\RR^4)$ with support in $\Region_0$,  
we define a corresponding loop function 
$x \mapsto l_{s, \gamma}(x) \doteq \int_0^1 dt \ 
s(x - \gamma(t)) \, \DOT{\gamma}(t)$, where $\DOT{\gamma}$
denotes the derivative of $\gamma$. One easily checks that 
$l_{s, \gamma} \in \Cc_1(\RR^4)$ and that it has support in 
$(\Region_0 + \gamma) \subset \Loop$. Moreover, if $\int \! dx \, s(x) \neq 0$
there is no $f \in \Dc_2(\RR^4)$ 
with support in $\Loop$ such that $l_{s, \gamma} = \delta f$,
\ie \ $l_{s, \gamma}$ is co-closed but not co-exact in this region. 
As has been pointed out 
in \cite{BuCiRuVa}, this fact is crucial for the appearance 
of non-trivial topological charges. We will restrict attention 
here to charges of this particular type. 

The following lemma shows that for given loop-shaped region 
$\Loop$ and any function $g \in \Cc_1(\RR^4)$, having
support in $\Loop$, one finds within its co-cohomology class
(referring to the co-derivative) loop functions $l_{s, \gamma}$, as 
given above. 

\begin{lemma} \label{l3.1}
Let $g \in \Cc_1(\RR^4)$ have support in 
a loop-shaped region 
$\Loop$. There exist some loop function $l_{s, \gamma} \in \Cc_1(\RR^4)$ 
and some $f \in \Dc_2(\RR^4)$, both having their support 
in $\Loop$, such that $g = l_{s, \gamma} + \delta f$. 
Hence $g$ and $l_{s, \gamma}$ lie in the same co-cohomology class
relative to $\Loop$. 
\end{lemma}

\begin{proof}
This statement is equivalent to the statement that 
$\star g = \star l_{s, \gamma} + d \star f$, where $\star$ is the 
Hodge operator. In other words, the 
cohomology classes of these functions must be related by  
$[\star g] = [\star l_{s, \gamma}] \in H_{c}^3(\Loop)$, where 
$H_{c}^3(\Loop)$ denotes the third compact de~Rham cohomology group
of forms which are compactly supported 
in the given open region. Making use of 
Poincar\'e duality \cite[Thm.~5.12]{GrHaVa} and 
then of the fact that $\Loop$ is 
homotopic to $\gamma$ which, being simple, is homotopic to $S^1$, 
we get 
\begin{equation} \label{deRham}
 H_{c}^3(\Loop)^* \approx H^1(\Loop)
\approx H^1(\gamma) \approx H^1(S^1) \approx \RR \, . 
\end{equation}
Here $H_{c}^3(\Loop)^*$ stands for the algebraic dual of $H_{c}^3(\Loop)$, 
\ $H^1(\, \cdot \,)$ denotes the first cohomology group of the 
respective regions, and the symbol \ $\approx$ 
indicates isomorphisms between the cohomology groups. 
Since $ H_{c}^3(\Loop)^*$ is finite dimensional,
one has  $ H_{c}^3(\Loop)^* \approx  H_{c}^3(\Loop)$ which
implies $H_{c}^3(\Loop) \approx \RR$. 

Now let $l_{s, \gamma}$ be any of the loop functions constructed above. Then
$\star l_{s, \gamma}$ has support in $\Loop$ and 
$d \star l_{s, \gamma} = - \star \delta l_{s,\gamma} = 0$  
since $ l_{s, \gamma}$ is co-closed. Moreover,  $\star l_{s, \gamma}$
is exact iff $ l_{s, \gamma}$ is co-exact which is the case
iff $\int \! dx \, s(x) = 0$ for the underlying
scalar function~$s$, cf.~\cite[\S 1]{Roberts} and
the subsequent section for a proof. 
Since the possible values of $\int \! dx \, s(x)$ \  
exhaust $\RR$ and $H_{c}^3(\Loop) \approx \RR$, 
it follows that within the 
cohomology class of any given $\star g$ one can find 
the Hodge dual of a loop function, $\star l_{s, \gamma}$,
proving the statement.
\end{proof}

The following statement concerning the forms $\sigma_\pi$ in
\eqref{symplectic} is a consequence of the preceding 
result and the causal Poincar\'e lemma, established in 
the appendix of \cite{BuCiRuVa}. 

\begin{lemma} \label{l3.2}
Let $\Loop_1$, $\Loop_2$ be two spacelike separated 
loop-shaped regions and let $g_1, g_2 \in \Cc_1(\RR^4)$ have support in 
$\Loop_1$, respectively $\Loop_2$. 
Moreover, let $\sigma_\pi$ be linear in both entries. 
There are loop functions 
$l_{s_1, \gamma_1}, \, l_{s_2, \gamma_2} \in \Cc_1(\RR^4)$ 
in the co-cohomology classes of $g_1$, $g_2$ relative to 
$\Loop_1$, respectively $\Loop_2$, 
such that $\sigma_\pi(g_1, g_2) = \sigma_\pi(l_{s_1, \gamma_1}, l_{s_2, \gamma_2})$. 
If $g_1$ or $g_2$ belong to the trivial co-cohomology
class, this expression vanishes.
\end{lemma}
\begin{proof}
By the preceding lemma there exists for given $g_1 \in \Cc_1(\RR^4)$,
having support in $\Loop_1$, a loop function  
$l_{s_1, \gamma_1} \in \Cc_1(\RR^4)$ and some $f_1 \in \Dc_2(\RR^4)$,
both having support in that region, such that 
$g_1 = l_{s_1, \gamma_1} + \delta f_1$.
By a partition of unity we split the function $f_1$ into a sum of 
test functions $f_{1, m} \in \Dc_2(\RR^4)$ having their supports 
in double cones $\Region_{1,m}$ in the spacelike complement of 
$\Loop_2$, \ $m = 1, \dots, M$, such that 
$f_1 = \sum_{m=1}^M f_{1,m}$. 

Now by the causal Poincar\'e lemma \cite{BuCiRuVa}
there exists for the given $g_2$ and each 
$m = 1, \dots , M$ a function
$f_{2,m} \in \Dc_2(\RR^4)$ that has compact support in the
spacelike complement of $\Region_{1,m}$ and satisfies $g_2 = \delta f_{2,m}$.
By another partition of unity 
one can proceed from $f_{2,m}$ to functions $f_{2,m,n} \in \Dc_2(\RR^4)$,
having support in double cones $\Region_{2,m,n}$ 
which are spacelike separated from the double cone $\Region_{1,m}$, 
$n = 1, \dots , N$, and  
$\sum_{n=1}^N f_{2,m,n} = f_{2,m}$ for $m = 1, \dots , M$.
After these preparations it follows from the postulated 
linearity properties of the form $\sigma_\pi$ and relation 
\eqref{a2} that 
$$
\sigma_\pi(g_1 - l_{s_1, \gamma_1}, g_2) =
\sum_{m=1}^M \sum_{n=1}^N \sigma_\pi(\delta f_{1,m}, \delta f_{2,m,n}) = 0 \, . 
$$
Thus $\sigma_\pi(g_1, g_2) = \sigma(l_{s_1, \gamma_1}, g_2)$, hence also 
$\sigma_\pi(g_2, l_{s_1, \gamma_1}) = \sigma_\pi(l_{s_2, \gamma_2}, l_{s_1, \gamma_1})$,
where $l_{s_2, \gamma_2}$ is any loop function in the co-cohomology class of $g_2$. 
Since $\sigma_\pi$ is skew symmetric, the first part of the 
statement follows. The second part is a consequence of the fact
that one may choose $l_{s_1, \gamma_1} = 0$ if $g_1$ belongs to the 
trivial co-cohomology class, and similarly for~$g_2$. 
\end{proof}

This lemma shows that the value of the form $\sigma_\pi(g_1, g_2)$ 
depends only on the co-cohomology classes of $g_1, g_2$ and the 
deformation retracts $\gamma_1, \gamma_2$ of the loop-shaped regions 
$\Loop_1, \Loop_2$, both being encoded in the loop functions 
$l_{s_1, \gamma_1}, \, l_{s_2, \gamma_2}$. We recall that for given 
loop $\gamma$, the co-cohomology classes
are fixed by the \textit{class values} which are given by the 
integral $\kappa \doteq \int \! dx \, s(x) \in \RR$ of the scalar
functions $s$ entering in the definition of~$l_{s, \gamma} \, $; 
changing the sense of traversal of $\gamma$ 
results in a sign change of $\kappa$. 
The support properties of $s$ are constrained by the 
condition \ $(\mbox{supp} \, s + \gamma) \subset \Loop$ for given 
loop-shaped region~$\Loop$. Thus for 
given loops $\gamma_1, \gamma_2$, 
the expression $\sigma_\pi(l_{s_1, \gamma_1}, \, l_{s_2, \gamma_2})$ 
is proportional to the product 
$\kappa_1 \kappa_2$ of the class values of the underlying 
loop functions.

It also follows from this lemma that by deforming 
the loops $\gamma_1, \gamma_2$ in $l_{s_1, \gamma_1}, \, l_{s_2, \gamma_2}$ 
to neighbouring loops $\gamma_1^\prime , \gamma_2^\prime$,
whilst keeping the product of their class values fixed, \ie 
\mbox{$\kappa_1^\prime \kappa_2^\prime = \kappa_1 \kappa_2$}, 
one can proceed to loop functions 
$l_{s_1^\prime, \gamma_1^\prime}, \, l_{s_2^\prime, \gamma_2^\prime}$ without changing the 
value of the symplectic form, 
$\sigma_\pi(l_{s_1, \gamma_1}, l_{s_2, \gamma_2}) = 
\sigma_\pi(l_{s_1^\prime, \gamma_1^\prime}, \, l_{s_2^\prime, \gamma_2^\prime})$. 
The latter loops are in turn deformation 
retracts of neighbouring loop-shaped regions $\Loop_1^\prime, \Loop_2^\prime$. 
Iterating this procedure, one can 
deform the spacelike loops $\gamma_1, \gamma_2$ 
to disjoint simple loops $\beta_1, \beta_2$ in the 
time zero plane~$\RR^3$, 
keeping the value of the symplectic form fixed,
\mbox{$\sigma_\pi(l_{s_1, \gamma_1}, l_{s_2, \gamma_2}) = 
\sigma_\pi(l_{s_1^\prime, \beta_1}, l_{s_2^\prime, \beta_2})$}. 
This deformation can be accomplished by 
jointly scaling the time coordinates of the loops to zero.
In the next step we show that
for fixed product $\kappa_1 \kappa_2$
the latter expression depends only on the 
homology class of $\beta_1$ in 
$\, \RR^3 \backslash \mbox{supp} \, \beta_2$. 

\begin{lemma} \label{l3.3}
Let $l_{s_1, \beta_1}, \, l_{s_2, \beta_2}$ be loop functions,
having spacelike separated supports,
which are assigned to disjoint linked loops 
$\beta_1, \beta_2 \subset \RR^3$. 
Moreover, let 
$\widetilde{\beta}_1 \subset \RR^3 \backslash \mathrm{supp} \, \beta_2$ 
be in the same homology class as $\beta_1$ with regard to 
$\RR^3 \backslash \mathrm{supp} \, \beta_2$. 
There exists a loop function $l_{\widetilde{s}_1, \widetilde{\beta}_1}$
in the co-cohomology class of $l_{s_1, \beta_1}$ relative  
to the spacelike complement of the support of 
$l_{s_2, \beta_2}$, such that 
$\sigma_\pi(l_{s_1, \beta_1}, l_{s_2, \beta_2}) = 
\sigma_\pi(l_{\widetilde{s}_1, \widetilde{\beta}_1}, l_{s_2, \beta_2})$. 
\end{lemma}

\noindent 
{\bf Remark.} Since the notion of homology is weaker than homotopy, 
this lemma applies 
in particular to loops which are homotopy equivalent. 

\begin{proof}
According to the definition of homology equivalence,  
cf.~\cite[Ch.~2.1]{Hat}, there are singular surfaces 
$\Surf_k \subset \RR^3 \backslash \mbox{supp} \, 
\beta_2$ and integers $m_k \in \ZZ$, $k = 1, \dots , n$, such that 
$(\beta_1 - \widetilde{\beta}_1) = 
\sum_k \, m_k \, \partial \Surf_k$, where $\partial \Surf_k$
denotes the boundary of $\Surf_k$. Given the class value
$\kappa_1$ of $l_{s_1, \gamma_1}$, one picks a scalar function
$\widetilde{s}_1 \in \Dc_0(\RR^4)$ with 
$\int \! dx \, \widetilde{s}_1(x) = \kappa_1$ such that
$(\mbox{supp} \, \widetilde{s}_1 + \Surf_k)$ lies in the spacelike complement
of the support of $l_{s_2, \gamma_2}$ for all $k$. Note that a replacement
of the scalar function $s_1$ in the definition of 
$l_{s_1 , \gamma_1}$ by $\widetilde{s}_1$ does not change 
its co-cohomology class and the value of 
$\sigma_\pi(l_{s_1, \gamma_1}, l_{s_2 , \gamma_2})$, c.f.~Lemma \ref{l3.2}. 

With this input one defines the function 
$f \doteq \sum_k m_k f_k \in \Dc_2(\RR^4)$, having support in 
the spacelike complement of the support of $l_{s_2 , \gamma_2}$,
where,   
\begin{align*}
& x \mapsto f_k^{\mu \nu}(x) \doteq \\
& 
\int_0^1 \! dt \int_0^{1-t} \! du \ \widetilde{s}_1(x - \Surf_k(t,u)) \,
\big( \partial_t \Surf_k(t,u)^\mu \, \partial_u \Surf_k(t,u)^\nu
- \partial_t \Surf_k(t,u)^\nu \partial_u \Surf_k(t,u)^\nu \big) \, .
\end{align*}
By a routine computation, one obtains  
$\delta f_k = l_{\widetilde{s}_1, \partial \Surf_k}$ for $k = 1, \dots , n$,
giving 
$$
\delta f = \sum_k m_k \, \delta f_k = \sum_k m_k \, 
l_{\widetilde{s}_1, \partial \Surf_k} = 
l_{\widetilde{s}_1, \beta_1}  - l_{\widetilde{s}_1, \widetilde{\beta}_1} \, .
$$
Decomposing $f$ into a sum of test functions having their 
supports in double cones in the spacelike complement of the 
support of~$l_{s_2, \beta_2}$, it follows from Lemma \ref{l3.2} and
the argument given there that  
$$
\sigma_\pi(l_{s_1, \beta_1}, l_{s_2, \beta_2}) - 
\sigma_\pi(l_{\widetilde{s}_1, \widetilde{\beta_1}}, l_{s_2, \beta_2})
= \sigma_\pi(l_{\widetilde{s}_1, \beta_1} - 
l_{\widetilde{s}_1, \widetilde{\beta}_1}, l_{s_2, \beta_2})
= \sigma_\pi(\delta f,  l_{s_2, \beta_2}) = 0 \, ,
$$
completing the proof.  
\end{proof}

We have now the information needed 
to prove the main result of this section. 

\begin{proposition}
Let $\Loop_1$, $\Loop_2$ be two spacelike separated   
loop-shaped regions which can continuously be retracted to 
spacelike linked loops $\gamma_1$ and $\gamma_2$, respectively. 
Moreover, let $\sigma_\pi$ be linear in both entries. 
Then, for any $g_1, g_2 \in \Cc_1(\RR^4)$ having support in 
$\Loop_1$, respectively  $\Loop_2$, one has 
$\sigma_\pi(g_1, g_2) = 0$. 
\end{proposition}
\begin{proof}
According to Lemma~\ref{l3.2} there exist loop 
functions $l_{s_1, \gamma_1}$, $l_{s_2, \gamma_2}$ with corresponding 
class values $\kappa_1$, $\kappa_2$ such that 
 $\sigma_\pi(g_1, g_2) = \sigma_\pi(l_{s_1, \gamma_1}, l_{s_2, \gamma_2})$. 
Disregarding the trivial case, where one of the class
values is zero, we cancel the inherent product 
$\kappa_1 \kappa_2$ and assume that $\kappa_1 = \kappa_2 = 1$.
Moreover, we assume 
that the supports of the scalar smearing functions $s_1, s_2$
are chosen such that the resulting loop functions always have
spacelike separated supports. 

Making use again of 
Lemma~\ref{l3.2}, one can deform the spacelike separated 
loops $\gamma_1, \gamma_2 \subset \RR^4$ to disjoint loops 
$\beta_1, \beta_2 \subset \RR^3$ without changing
the value of the symplectic form, \ie 
$\sigma_\pi(g_1, g_2) = \sigma_\pi(l_{s_1, \beta_1}, l_{s_2, \beta_2})$. 
By Lemma~\ref{l3.3} one can then proceed from
the loop $\beta_1$ to any other loop 
$\widetilde{\beta}_1$ in the same homology class 
within the region $\RR^3 \backslash \mbox{supp} \, \beta_2$,
retaining the value 
of the form, $\sigma_\pi(l_{s_1, \beta_1}, l_{s_2, \beta_2}) =
\sigma_\pi(l_{{s}_1, \widetilde{\beta}_1} , l_{s_2, \beta_2})$. 
Similarly, one can replace the loop 
$\beta_2$ by any other loop $\widetilde{\beta}_2$
in the same homology class 
relative to $\RR^3 \backslash \mbox{supp} \, \widetilde{\beta}_1$.
Taking into account that $\sigma_\pi$ is skew symmetric, this gives 
altogether 
$$
\sigma_\pi(g_1, g_2)
= \sigma_\pi( l_{{s}_1, \widetilde{\beta}_1}, l_{s_2, \beta_2}) 
= - \sigma_\pi( l_{s_2, \beta_2}, l_{{s}_1, \widetilde{\beta}_1})
= - \sigma_\pi( l_{{s}_2, \widetilde{\beta}_2}, 
l_{{s}_1, \widetilde{\beta}_1})
= \sigma_\pi( l_{{s}_1, \widetilde{\beta}_1}, 
l_{{s}_2, \widetilde{\beta}_2}) \, .
$$

Now the initial curves 
${\beta}_1, {\beta}_2 \subset \RR^3$
are simple, so their supports are homeomorphic to 
the circle $S^1$.  For any such curve $\beta$, 
one has for the homology group   
$H_1(\RR^3 \backslash \beta ) \approx H^1( \beta) \approx \ZZ$ 
by Alexander duality \cite[Ch.~VI, Cor.~8.6]{Br}. 
(Note that we are dealing here with homology 
groups having coefficients in $\ZZ$.)  
Hence one can choose for 
$\widetilde{\beta}_1, \widetilde{\beta}_2$
linked circles ${\zeta}_1, {\zeta}_2$
of equal radius (forming a Hopf link) which, by convention,  
are traversed in positive direction in their respective planes;
depending on the homology classes of the loops 
$\widetilde{\beta}_1, \widetilde{\beta}_2$, these circles 
are traversed $n_1$, respectively $n_2$ times, $n_1, n_2 \in \ZZ$,
giving 
$\sigma_\pi(l_{{s}_1, \widetilde{\beta}_1}, 
l_{{s}_2, \widetilde{\beta}_2})
= n_1 n_2 \, \sigma_\pi(l_{s_1, \zeta_1}, l_{s_2, \zeta_2})$.
Finally, one continuously exchanges the pair of circles
$\zeta_1, \zeta_2$
into $\zeta_2, \zeta_1$
while keeping their distance greater than zero.  
By another application of Lemma \ref{l3.3}, one then obtains 
$\sigma_\pi(l_{s_1, \zeta_1}, l_{s_2, \zeta_2})
= \sigma_\pi(l_{s_2, \zeta_2}, l_{s_1, \zeta_1})
= 0$, where the second equality follows from the skew symmetry 
of $\sigma_\pi$. Since 
$\sigma_\pi(g_1, g_2) = n_1 n_2 \, \sigma_\pi(l_{s_1, \zeta_1}, l_{s_2 \zeta_2})$,
this completes the proof of the statement. 
\end{proof}

So we conclude that the topological charges exhibited in \cite{BuCiRuVa}
are trivial in any regular representation of the universal algebra 
in which the corresponding intrinsic vector potential depends linearly on 
the elements of $\Cc_1(\RR^4)$. 

\section{A vacuum state with non-trivial topological charge
\label{go-go}}
\setcounter{equation}{0}

In this section we present a regular vacuum representation of the universal 
algebra with a non-trivial topological charge. This simple example derives 
from the free electromagnetic field. Roughly speaking,
we merge the electric and magnetic parts of this linear 
field in a non-linear manner,
thereby obtaining a spacelike linear field carrying a topological charge.
The resulting Haag-Kastler net coincides with the original net, so our 
construction relies on a re-interpretation of the original 
theory.  It shows that the universal 
algebra gives leeway to the appearance of topological charges.

We proceed from the regular vacuum state $\omega_0$ on the universal
algebra, giving rise to a vanishing electric current and a linear 
free vector potential \cite{BuCiRuVa}. Its GNS representation is denoted by 
$(\pi_0, \Hil_0, \Omega_0)$ and the resulting vector potential by $A_0$ 
with domain $\Dom_0 \subset \Hil_0$. 
As is well known, this special vacuum state is quasi-free, so 
the correlation functions of the potential are 
fixed by the two-point function
$$
\langle \Omega_0, A_0(g_1) A_0(g_2) \Omega_0 \rangle 
= - (2 \pi)^{-3} \! \int \! dp \, \theta(p_0) \delta(p^2) \, 
\widehat{g}_1(-p) \widehat{g}_2(p) \, , \ \ 
g_1, g_2 \in \Cc_1(\RR^4) \, . 
$$
Here $\widehat{g}$ denotes the Fourier transform of $g$ and the product
of functions is defined by the Lorentz scalar product of their components. 
We recall that there exists on $\Hil_0$ a continuous unitary representation
$U_0$ of the Poincar\'e group $\Poin_+^\uparrow$ that satisfies the
relativistic spectrum condition, leaves $\Omega_0$ 
invariant and induces Poincar\'e transformations of the vector potential, 
$U_0(P) A_0(g) U_0(P)^{-1} = A_0(g_P)$ for $P \in \Poin_+^\uparrow$. 

For the construction of a potential which is associated with a 
topological charge we 
make use of the following facts. Given $g \in \Cc_1(\RR^4)$, 
let $G \in \Dc_2(\RR^4)$ be any of its co-primitives, 
\ie $g = \delta G$. If $G, G^\prime$ are two such co-primitives
for the given $g$,
there exists a test function with values in totally
antisymmetric tensors of rank three, $k \in \Dc_3(\RR^4)$,
such that $G - G^\prime = \delta k$. Thus, using proper
coordinates, the tensor 
$$
\overline{G}^{\, \mu \nu} \doteq \int \! dx \ G^{\mu \nu}(x)
$$
depends only on $g$ and not on the chosen co-primitive $G$. 
Moreover, it is invariant under translations 
and transforms covariantly under Lorentz transformations 
of the underlying $g$.
Second, also the operators 
$A_0(\delta \star G)$ depend only on $g \in \Cc_1(\RR^4)$, but not on the
chosen co-primitive $G \in \Dc_2(\RR^4)$. This follows from 
$$
A_0(\delta \star \delta k) = - A_0(\delta d \star k) = - j_0(\star k) = 0 \, , 
\quad k \in \Dc_3(\RR^4) \, ,$$
bearing in mind that the electric current $j_0$ vanishes 
according to our assumptions. 

Now let $\theta_\pm$ be the characteristic functions of the positive,
respectively negative reals (where $0$ is excluded) and let, 
in proper coordinates,  
$\overline{G}^2 \doteq \overline{G}_{\mu \nu} \, \overline{G}^{\mu \nu}$.  
Given any~$g \in \Cc_1(\RR^4)$ we decompose it into  
functions having their supports in disjoint connected
regions, $g = \sum_n g_n$, and define on the domain $\Dom_0$ a 
\textit{topological potential}, putting 
\begin{equation} \label{newpotential}
A_T(g) \doteq \sum_n \big( \theta_+(\overline{G_n}^2) \, A_0(\delta G_n) 
+ \theta_-(\overline{G_n}^2) \, A_0(\delta \star G_n) \big)  \, .
\end{equation}
This definition is clearly meaningful for test functions 
$g$ which can be 
split into a finite sum of functions having disjoint 
connected supports. But since $A_0$ is an operator-valued distribution, 
it is not difficult to show that this definition can be extended 
by continuity to arbitrary test functions. 

Omitting this step, let us turn to the proof that the operators 
$A_T(g)$, $g \in \Cc_1(\RR^4)$, arise
as vector potentials of some regular vacuum representation of the universal 
algebra that carries a non-trivial topological charge. 
We begin by recalling that 
the domain $\Dom_0$ is a common core for the operators $A_0(g)$,
$g \in \Cc_1(\RR^4)$, which is stable under the action of the
corresponding exponential functions (Weyl operators). Since 
the step from $A_0$ to $A_T$ is based on some non-linear transformation
of the test functions, c.f.\ below, it follows that 
the operators $A_T(g)$ share this domain property. In particular, 
their exponentials \ $e^{i a A_T(g)}$ satisfy condition 
\eqref{a1} for $a \in \RR$, $g \in \Cc_1(\RR^4)$. 

For the proof that $A_T$ is spacelike linear, \ie 
condition \eqref{a2} holds for its exponentials, we note 
that for any pair $g_1, g_2 \in \Cc_1(\RR^4)$ satisfying  
$\mbox{supp} \, g_1 \Times \mbox{supp} \, g_2$, all
functions $g_{1,m}$, $g_{2,n}$ appearing in their 
respective decompositions satisfy this condition as well. 
Now given functions $g_{1,m}, g_{2,n} \in \Cc_1(\RR^4)$ with
$\mbox{supp} \, g_{1,m} \Times \mbox{supp} \, g_{2,n}$, 
it follows from the local  
Poincar\'e lemma \cite{BuCiRuVa} that there are corresponding co-primitives 
$G_{1,m}, G_{2,n} \in \Dc_2(\RR^4)$ satisfying 
$\mbox{supp} \, G_{1,m} \Times \mbox{supp} \, G_{2,n}$. 
Since the initial potential $A_0$ satisfies condition 
\eqref{a2} and since the Hodge operator~$\star$ as well as the 
co-derivative $\delta$ do not impair the localisation properties of 
test functions, it follows from the defining relation 
\eqref{newpotential} that $[A_T(g_1), A_T(g_2)] = 0$ 
if $\mbox{supp} \, g_1 \Times \mbox{supp} \, g_2$, as claimed. 

Next, for given $g \in \Cc_1(\RR^4)$ with 
co-primitive $G \in \Dc_2(\RR^4)$ and any Poincar\'e transformation 
$P \in \Poin_+^\uparrow$, the transformed 
tensor $G_P$ is a co-primitive of the transformed $g_p$. Since 
$\overline{G_p}^2 = \overline{G}^2$ and disjoint sets are mapped to
disjoint sets by Poincar\'e transformations, it follows from 
relation \eqref{newpotential} that \ 
$U_0(P) A_T(g) U_0(P)^{-1} = A_T(g_P)$ \ for $P \in \Poin_+^\uparrow$ 
and $g \in \Cc_1(\RR^4)$. Thus the topological potential $A_T$ transforms 
covariantly under the given representation of the Poincar\'e 
group and the interpretation of the vector $\Omega_0$ 
as vacuum state does not change. 
What changes, however, are the expectation
values of the elements of the universal algebra. As has been
explained in \cite{BuCiRuVa}, these are fixed by the generating
function of the vector potential. 

In the case at hand, the generating function 
can easily be obtained from the one of the free field. 
To explicate this we introduce the notation 
$$
G_T \doteq \sum_n \big( \theta_+(\overline{G_n}^2) \, G_n
+ \theta_-(\overline{G_n}^2) \star G_n \big) \, ,
$$

\vspace*{-2mm} \noindent
where we recall that the $G_n \in \Dc_2(\RR^4)$ are co-primitives of 
the different disjoint portions $g_n$ of 
$g \in \Cc_1(\RR^4)$, $n \in \NN$. Since $A_0$ depends linearly on 
test functions, this yields  the equality 
$A_T(g) = A_0(\delta G_T)$. 
Denoting the  
topological vacuum state on the universal algebra $\fV$ by~$\omega_T$,
its generating function is given by, $a \in \RR$ and $g \in \Cc_1(\RR^4)$, 
\begin{align} \label{generating}
& \omega_T(V(a,g)) & \nonumber \\ 
& \doteq \langle \Omega_0, 
e^{\mbox{\footnotesize  $ \, i a A_T(g)$ }} 
\Omega_0 \rangle = 
e^{\mbox{\footnotesize 
$- a^2 (2 \pi)^{-3} \! \int \! dp \, \theta(p_0) \delta(p^2) \,
( p \, \widehat{G_T}(-p) ) 
( p \, \widehat{G_T}(p) ) / 2 $ }} \, , &
\end{align}
where, in proper coordinates,  
$(p  \, \widehat{G_T}(p))_\mu = p^\nu \widehat{G_T}(p)_{\mu \nu}$.
Thus the GNS representation induced by $\omega_T$ 
acts on $\Hil_0$ and the cyclic GNS vector $\Omega_0$ represents 
the state. The corresponding homomorphism $\pi_T$, mapping
the elements of the universal algebra $\fV$ to  bounded
operators, is given by 
$\pi_T(V(a,g)) \doteq e^{\, i a A_0(\delta G_T)}$ 
for $a \in \RR$, $g \in \Cc_1(\RR^4)$.
Since the free field $A_0$ has c-number commutation relations,  
one obtains for the commutator of the topological potential 
\mbox{$[A_T(g_1), A_T(g_2)] = [A_0(\delta G_{1 T}), 
A_0(\delta G_{2 T})] \in \CC \, 1$}. So the
commutator lies in the centre of the represented algebra
and condition~\eqref{a3} is clearly satisfied.  

It remains to prove that the topological vacuum state 
$\omega_T$ gives rise to some non-trivial topological charge, 
\ie that there exist functions with spacelike separated, linked supports,  
$\mbox{supp} \, g_1 \perp \mbox{supp} \, g_2$,
for which the commutator of the topological vector 
potential has values different from zero. 
We restrict attention here to functions which both have their 
supports in some connected region. Then one obtains
\begin{align} \label{commutator}
[A_T(g_1), A_T(g_2)] & = 
\big( \theta_+(\overline{G_1}^2) \theta_+(\overline{G_2}^2) 
+ \theta_-(\overline{G_1}^2) \theta_-(\overline{G_2}^2) \big) \,
\Delta(G_1, G_2) \, 1 \nonumber \\
& + \big( \theta_+(\overline{G_1}^2) \theta_-(\overline{G_2}^2) 
- \theta_+(\overline{G_2}^2) \theta_-(\overline{G_1}^2) \big) \,
\Delta(G_1, \star G_2) \, 1 \, .
\end{align}
Here $\Delta$ is the commutator function of the free Maxwell field,
\begin{equation} \label{pauli-jordan}
\Delta(G_1, G_2) \doteq (2 \pi)^{-3} 
\int \! dp \, \epsilon(p_0) \delta(p^2) \,
\big( p \, \widehat{G_1}( -p) \big) \! \big( p \, \widehat{G_2}(p) \big) \, ,
\end{equation}
and we made use of the fact that the free field and its Hodge dual
have the same commutator function, 
$\Delta(\star \, G_1, \star \, G_2) = \Delta(G_1, G_2)$. 
Moreover, since \mbox{$\star \star \upharpoonright \Dc_2(\RR^4) = -1$}, we
have $\Delta(G_2, \star \, G_1) = - \Delta(\star \, G_1, G_2)
= \Delta(\star \, G_1, \star \! \star G_2) = \Delta(G_1, \star \, G_2)$,
\ie the latter distribution is symmetric in $G_1, G_2$. 

Relying on results by Roberts \cite{Roberts}, we will make use of the 
appearance of the distribution $\Delta(G_1, \star \, G_2)$
in the commutator of the topological field. Roberts has
shown that this term has non-trivial values for certain pairs of 
test functions $g_1, g_2 \in \Cc_1(\RR^4)$ with spacelike 
separated supports. But we have to 
control here also the connectivity of their supports
as well as the pre-factor of the Roberts term 
appearing in relation~\eqref{commutator}. 

Choosing proper coordinates, let
$[0,1] \ni t \mapsto \gamma^{(i k)} (t)$ be unit circles 
which are centred at time zero at the origin of the $i$-$k$-plane, 
$i,k \in \{ 1,2,3 \}$. We consider the corresponding loop functions 
$x \mapsto g^{(ik)}(x) \doteq \int_0^1 \! dt \, s(x - \gamma^{(i k)}(t)) 
\DOT{\gamma}{}^{(i k)}(t)$, where $s$ is a scalar test function
with support in a small neighbourhood of $0 \in \RR^4$. 
It is convenient to choose functions of the form 
$x \mapsto s^{(i k l)}(x) \doteq a(x_0,x_l) \, b(x_i^2 + x_k^2)$,
$i,k,l \in \{ 1,2,3 \}$. Then one obtains 
$ x \mapsto g^{(i k)}_\mu(x) = ( -\delta_{\mu i} \, x_k + \delta_{\mu k} \, x_i )
\ a(x_0, x_l) \, c(x_i^2 + x_k^2)$, where 
$$
r \mapsto c(r^2) \doteq |r|^{-1} \! \int_0^1 \! dt \, \cos(2 \pi \, t) 
\, b(r^2 - 2r \cos(2\pi  \, t) + 1)
\, .
$$
Note that the latter function has support in small neighbourhoods  
of $r = \pm 1$
since~$b$ has support in a small neighbourhood of the origin.
Thus by choosing for $a, b$ test functions of
definite sign which have their supports in small 
connected regions, the supports of the resulting functions 
$x \mapsto g^{( i, k )}(x)$ are contained in loop-shaped regions 
(solid tori). 
Moreover, these functions have co-primitives of the simple form 
$$
x \mapsto G^{(ik)}_{\mu \nu} (x) = 2^{-1} (\delta_{\mu k} \delta_{\nu i} -
\delta_{\nu k} \delta_{\mu i}) \, a(x_0, x_l) \, 
C(x_i^2 + x_k^2) \, ,
$$
where $r \mapsto C(r^2) \doteq - \int_{r^2}^\infty \! du \, c(u) $ 
is a test function of compact support which is constant for small
argument.   

Whereas these (suitably shifted) co-primitives give rise 
to non-trivial values of the Roberts term in relation \eqref{commutator}, 
there arises a problem with its pre-factor. Since the tensors
$G^{(i k)}$ have only spatial components, \ie are of magnetic type, 
one has $\overline{G^{(i k)}}^2 < 0$ and the pre-factor vanishes. 
This problem can be solved by proceeding from the loop functions 
$g^{(ik)}$ to the family of functions $(g^{(i k)} + \eta \, g^{(0 l)})$,
$\eta \in \RR$. Here the $g^{(0 l)} \in \Cc_1(\RR^4)$ may be chosen
according to 
$$
x \mapsto g^{(0 l)}_\mu (x) = 
(\delta_{\mu 0} \, \partial_l - \delta_{\mu l} \, \partial_0) \, 
a(x_0, x_l) c(x_i^2 + c_k^2) \, .
$$
Then the functions $(g^{(i k)} + \eta \, g^{(0 l)})$, $\eta \in \RR$, 
have the same support
as the loop functions. The co-primitives of $g^{(0 l)}$ are of electric 
type and are given by
$$
x \mapsto G^{(0 l)}_{\mu \nu} = 
(\delta_{\mu 0} \delta_{\nu l} - \delta_{\nu 0} \delta_{\mu l}) 
a(x_0, x_l) c(x_i^2 + c_k^2) \, . 
$$
So one obtains 
$\overline{(G^{(i k)} + \eta G^{(0 l)})}^2 =  \overline{G^{(i k)}}^2
+ \eta^2 \overline{G^{(0 l)}}^2$ which is positive for sufficiently large
$\eta$. Modifying one of the loop functions in 
the Roberts term in this manner, one does not change its value since 
the co-primitives of electric type are mapped by the Hodge 
operator $\star$
to functions of magnetic type whose commutators vanish. 

For the proof that the commutator of the topological 
vector potential does not vanish,
we pick the modified loop function $(g^{(1 2)} + \eta \, g^{(0 3)})$
and the loop function $g^{(1 3)}_{\, e}$, where the latter function is translated 
into the $1$-direction by the unit vector $e$. So the supports of the 
resulting functions are intertwined, forming a Hopf link. Since 
$\overline{G^{(0 3)}}^2 < 0$ and, 
for suitable $\eta$, \ $\overline{(G^{(1 2)} + \eta \, G^{(0 3)})}^2 > 0$, 
it follows that \ 
$$
[A_T(g^{(1 2)} + \eta \, g^{(0 3)}), \, A_T(g^{(1 3)}_{\, e}) ]
= \Delta(G^{(1 2)}, \star \, G^{(1 3)}_{\, e}) \, 1 \, .
$$
Inserting the functions given above and 
disregarding numerical factors, one gets  
\begin{align*}
& \hspace*{3mm} \Delta(G^{(1 2)}, \star \, G^{(1 3)}_{\, e}) & \\
& \propto \int \! dp \, \varepsilon(p_0) \delta(p^2) \,
\big( \widehat{a}(-p_0,-p_3) \, \partial_1 \, 
\widehat{c} (p_1^2 + p_2^2) \big)
\ \big( p_0 \, \widehat{a}(p_0, p_2) \, 
\widehat{C}(p_1^2 + p_3^2) \big) \, e^{-i p e}  \, . &
\end{align*}
We make use now of the freedom to choose the function $a$ within the above 
limitations and proceed to the limit case where $a$ is the Dirac 
$\delta$ function. Then $\widehat{a}$ is constant and one can perform
the $p_0$-integration, turning the above integral into 
$$
\int \! d^3 \bip \, \big( \partial_1 \, \widehat{c} (p_1^2 + p_2^2) \big) \
\widehat{C}(p_1^2 + p_3^2) \, e^{-i p_1 e_1}
\propto \int \! dx_1 \, x_1 c(x_1^2) \ C((x_1 - e_1))^2 \, .
$$
Now the function $x_1 \mapsto 2 x_1 c(x_1^2)$ has support in an  
$\epsilon$-neighbourhood of $\pm 1$, 
whereas its primitive $x_1 \mapsto C(x_1^2)$
has a constant, non-zero value in the interval 
$| x_1 | \leq (1- \epsilon)$ and vanishes for 
$| x_1 | \geq (1+ \epsilon)$. Hence, due to the translation
$e_1$,  only positive values of~$x_1$ contribute to the 
above integral which consequently does not vanish. 
This shows that the Roberts term is different from
zero in the limit where $a$ is the Dirac delta function.  
In view of the continuity properties of the 
defining integral this is therefore also true for 
suitable approximating test functions $a$. Thus we have established 
the existence of topological charges in this model.

\begin{proposition}
Let $\omega_T$ be the regular vacuum state
on the universal algebra~$\fV$ which is fixed by the generating
function \eqref{generating}.
There exist functions $g_1, g_2 \in \Cc_1(\RR^4)$
with $\mathrm{supp}\, g_1 \perp \mathrm{supp} \, g_2$
such that the central group-theoretic commutator in condition~\eqref{a3} 
has values different from $1$ in the GNS 
representation induced by $\omega_T$. 
Thus there appear non-trivial topological charges in this
representation. 
\end{proposition}

We conclude this section by coming back to a remark made in the 
proof of Lemma~\ref{l3.1}. There it was stated that 
loop functions with support in loop-shaped regions, being co-closed, are 
in general not co-exact in these regions. This is evident now
from the preceding computations. For if 
the loop functions $g^{(i k)}$, having support
in some solid torus,  would have 
co-primitives $G^{(i k)}$ which have their support 
in the same region, the Roberts
term $\Delta(G^{(1 2)}, \star \, G^{(1 3)}_{\, e})$   
would vanish. This follows from the fact that 
the Hodge operator $\star$ does not
change supports and the commutator function~$\Delta$ vanishes
on test functions in $\Dc_2(\RR^4)$ with spacelike
separated supports. But this would be in conflict with the
results of the preceding computations.

\section{Electric currents and spacelike linear vector potentials
\label{current}}
\setcounter{equation}{0}

In the preceding section we have exhibited a theory with non-trivial 
topological charge, but trivial electric current. We will show now that
there exist topological charges in presence of given electric 
currents. This result may look surprising since, as is well known, it is 
in general not possible to couple any given current via field equations to 
a Wightman field, cf.~\cite{ArHaSch}. But, as we have shown, 
in the presence of topological charges 
the vector potentials can at best be spacelike linear, 
they are never Wightman fields. It turns out that 
within this larger class of fields, the inhomogeneous Maxwell 
equation does have solutions for almost any given 
current. Combining them with the
field of the preceding section, one obtains vector potentials carrying a 
topological charge also in the presence of electric currents.  

Our starting point is a conserved current $j(h)$, $h \in \Dc_1(\RR^4)$,
satisfying all Wightman axioms~\cite{StWi}. In particular, it is 
(real) linear on
test functions, local, and it transforms covariantly under 
a continuous unitary representation $U_J$ of the Poincar\'e group 
$\Poin_+^\uparrow$, satisfying the relativistic spectrum condition. 
We also assume that the current operators act on a common
stable core $\Dom_J \subset \Hil_J$ which contains the vacuum vector
$\Omega_J$ and which is stable under the action of 
the unitary operators $e^{i j(h)}$, $h \in \Dc_1(\RR^4)$. 
Examples of such currents are familiar from theories of 
charged free fields. 

It is our aim to show that there exists a spacelike linear vector 
potential $A_J$ which arises from a regular vacuum representation of 
the universal algebra $\fV$ on the Hilbert space $\Hil_J$ 
and which satisfies on the domain $\Dom_J$ the inhomogeneous
Maxwell equation $A_J(\delta d h) = j(h)$, $h \in \Dc_1(\RR^4)$. 
Note that $\delta d \Dc_1(\RR^4) \subset \Cc_1(\RR^4)$; so,
for given $j$, the potential 
$A_J$ is already defined by this equation on this subspace.
Thus we are faced with the problem to extend it consistently to all of 
$\Cc_1(\RR^4)$ and, less trivial, 
to determine its localisation properties from those of the given $j$. 
This is accomplished with the help of the following lemma.

\begin{lemma} \label{l5.1}
Let $g = \delta d h \in \delta d \Dc_1(\RR^4)$. 
(i) The pre-image $h \in \Dc_1(\RR^4)$ of $g$ is uniquely 
determined by g up to elements of \ $d \Dc_0(\RR^4)$.  
(ii) Let $\Region \supset \mathrm{supp} \, g$
be such that the complement of 
$\widetilde{\Region} \doteq (\Region + V_+) \bigcap (\Region + V_-)$
has trivial first homology, 
$H_1(\RR^4 \backslash \widetilde{\Region}) = \{ 0 \}$. 
There exist pre-images $h \in \Dc_1(\RR^4)$ of $g$ having support in 
any given neighbourhood of~$\widetilde{\Region}$. 
(iii) The map from $\delta d \Dc_1(\RR^4)$   
into the classes $\Dc_1(\RR^4) /d \Dc_0(\RR^4)$ of 
pre-images is continuous.
\end{lemma} 
\begin{proof}
(i) Let $h_1, h_2 \in \Dc_1(\RR^4)$ such that 
$ \delta d h_1 = g = \delta dh_2$. Then 
$$
\square \, d(h_1 - h_2) = (d \delta + \delta d) d (h_1 - h_2) = 0 \, ,
$$
\ie $d(h_1 - h_2)$ is a solution of the wave equation with compact
support and hence vanishes everywhere, $d(h_1 - h_2) = 0$. Thus, by the 
compact Poincar\'e lemma, there exists a scalar test function 
$s \in \Dc_0(\RR^4)$ such that $h_1 - h_2 = d s$. \

(ii) By assumption,    
$\square \, d h = (\delta d + d \delta) d h = d g$ and
consequently $d h = E_\pm d g + f_\pm$, where 
$E_\pm$ are the retarded and advanced Green's functions  
and $f_\pm$ are solutions of the 
wave equation. Since $g, h$ have compact support,
the support of $E_\pm \, dg$ is contained in future, respectively past
directed lightcones and this is therefore also true for 
$f_\pm$. Since the latter functions are solutions of the
wave equation, it follows that \mbox{$f_\pm = 0$}. 
Thus $dh =  E_\pm \, d g$ and consequently 
$ \mbox{supp} \ dh \subset \{\mbox{supp} \, g + V_+ \} \bigcap
\{\mbox{supp} \, g + V_- \} \subset \widetilde{\Region}$. 
Since $d h \upharpoonright \RR^4 \backslash \widetilde{\Region} = 0$
and $H_1(\RR^4 \backslash \widetilde{\Region}) = \{ 0 \}$, there 
is some smooth scalar function $s$ on that region, such that 
$(h - d s) \upharpoonright \RR^4 \backslash \widetilde{\Region} = 0$.
We choose now a smooth characteristic function $c$ which 
is equal to $1$ on $\widetilde{\Region}$ and vanishes outside 
some given neighbourhood of  $\widetilde{\Region}$. 
The function 
$\widetilde{h} \doteq (h - d (1-c) \, s)$ is then 
an element of $\Dc_1(\RR^4)$ which has support in this
neighbourhood and satisfies $d \widetilde{h} = g$, as claimed. 

(iii) Let $\Region_0 \subset \RR^4$ be any open double cone and let 
$g = \delta d h \in \delta d \Dc_1(\RR^4)$ have support 
in the interior of $\Region_0$. 
Since $\RR^4 \backslash \Region_0$ is 
simply connected and consequently  
$H_1(\RR^4 \backslash \Region_0) = \{ 0 \}$, we may assume 
according to the preceding step 
that $h$ has its support in  $\Region_0$ as well. Moreover, 
$d h = E_\pm dg = d E_\pm g$, hence 
$d E_\pm g \upharpoonright \RR^4 \backslash \Region_0 = 0$.  
The solutions of this equation are the
path integrals $x \mapsto \int_{x_a}^x \! dy_\mu (E_\pm g)^\mu (y)$
along any path from some
fixed initial point $x_a$ to $x$ within 
the region $\RR^4 \backslash \Region_0 $.
Picking a double cone $\Region$ containing $\Region_0$ in its interior
and a smooth characteristic function $c$ which is equal to $1$
on $\Region_0$ and vanishes in the complement of $\Region$,     
we put 
$$
x \mapsto h_\pm(x) \doteq 
E_\pm g(x) - d \, (1 - c(x)) \int_{x_a}^x \! dy_\mu (E_\pm g)^\mu (y) \, .
$$ 
These functions are elements of $\Dc_1(\RR^4)$, 
have support in $\Region$, and satisfy \mbox{$\delta d h_\pm = g$}.
Moreover, for fixed index $+$ or $-$, 
they depend continuously on the test 
functions $g \in \delta d \Dc_1(\RR^4)$ 
having support in $\Region_0$. Since the choice of 
regions was arbitrary, this completes  the proof. 
\end{proof}

We are equipped now with the tools to extend the vector 
potential $A_J$ to $\Cc_1(\RR^4)$. As explained above,
$A_J$ is already defined on the subspace $\delta d \Dc_1(\RR^4)$.
Elements of $\Cc_1(\RR^4)$ which belong to this subspace
will be marked  by an index $\smile$, that is $g_\smile \in  \Cc_1(\RR^4)$
means that there is some $h_\smile \in \Dc_1(\RR^4)$ such that 
$g_\smile = \delta d h_\smile$. 

For the desired extension of $A_J$ to $\Cc_1(\RR^4)$, we begin by considering 
the functions $g \in \Cc_1(\RR^4)$ whose support is connected.   
If $g = g_\smile$ we put $A_J(g) \doteq j(h_\smile)$,
otherwise we put $A_J(g) = 0$. This definition is consistent
according to part (i) of the preceding lemma 
since $j(d s) = 0$, $s \in \Dc_0(\RR^4)$, due to current conservation.  
Next, given any \mbox{$g \in \Cc_1(\RR^4)$}, we decompose it into functions
having disjoint connected supports, \mbox{$g = \sum_n g_n$}.
We then collect from this sum all terms $g_n$ which are 
of type $\smile$, giving $g_\smile = \sum_{\, n^\prime} g_{\, n^\prime \smile}$.
Denoting the remainder by $g_\frown$, 
this produces for any $g \in \Cc_1(\RR^4)$ a unique splitting 
$g = g_\smile + g_\frown$.
After these preparations we can extend the definition of 
the vector potential to
arbitrary test functions, putting on the dense domain $\Dom_J$
\begin{equation} \label{currentpotential}
A_J(g) \doteq A_J(\sum_{n^{\prime}} g_{n^\prime \smile} ) 
= j (\sum_{n^{\prime}} h_{n^\prime \smile} ) \, , 
\quad g \in \Cc_1(\RR^4) \, . 
\end{equation}
This definition is meaningful as it stands for test functions $g$ which can
be split into finite sums of functions with disjoint connected supports. 
Since $j$ is a Wightman field it can then be extended by continuity
to arbitrary test functions, relying on part (iii) of the
preceding lemma. With this construction, we can then write 
in an obvious notation $A_J(g) = A_J(g_\smile) = j(h_\smile)$,  
$g \in \Cc_1(\RR^4)$. In particular, $A_J(g_\frown) = 0$. 

Having defined the vector potential, let us turn now to the discussion 
of its properties. Since $A_J(g) = j(h_\smile)$ and
the domain $\Dom_J$ is a common core for the currents $j(h)$ 
which is stable under the action of the corresponding unitaries, 
$h \in \Dc_1(\RR^4)$, this is also true for $A_J(g)$ and the
unitaries $e^{ia A_J(g)}$, where $a \in \RR$, $g \in \Cc_1(\RR^4)$.
Thus condition~\eqref{a1} is satisfied. 

For the proof of 
condition \eqref{a2}, let $g_1, g_2 \in \Cc_1(\RR^4)$ such that
$\mbox{supp} \, g_1 \Times \mbox{supp} \, g_2$,
\ie $g_1, g_2$ have support in spacelike separated double cones. 
Then their $\smile$ contributions $g_{1 \smile}, g_{2 \smile}$ 
share this property 
and in view of the second part of Lemma \ref{l5.1} there exist corresponding 
pre-images $h_{1 \smile}, h_{2 \smile} \in \Dc_1(\RR^4)$ 
such that $\mbox{supp} \, h_{1 \smile} \Times \mbox{supp} \, h_{2 \smile}$. 
Since the current is 
linear in the test functions and local it follows that 
$$
e^{\mbox{\footnotesize $i a_1 A_J(g_1)$}} \, 
e^{\mbox{\footnotesize $i a_2 A_J(g_2)$}} = 
e^{\mbox{\footnotesize $i a_1 j(h_{1 \smile})$}} \, 
e^{\mbox{\footnotesize $i a_2 j(h_{2 \smile})$}} 
= e^{\mbox{\footnotesize $i \! j(a_1 h_{1 \smile} + a_2 h_{2 \smile} )$}} =
e^{\mbox{\footnotesize $i A_J(a_1 g_1 + a_2 g_2 )$}} \, ,
$$
in accordance with condition \eqref{a2}. 

The proof of condition~\eqref{a3} requires more work because of the 
intricate relation between the support of the elements of 
$\delta d \Dc_1(\RR^4)$ and of their pre-images in $\Dc_1(\RR^4)$. 
We begin with some technical lemma.

\begin{lemma}
Let $h_1, h_2 \in \Dc_1(\RR^4)$ such that \  
$\mathrm{supp} \, d \, h_1 \perp \mathrm{supp} \, h_2$.
Then one has \mbox{$[j(h_1), j(h_2)] = 0$}.
\end{lemma}
\begin{proof}
First, we assume that $h_2$ has support in some open double cone 
$\Region_2 \perp \mbox{supp} \, d \, h_1$. Since $\mbox{supp} \, d \, h_1$
is compact there exists another open double cone 
$\Region_1 \supset \Region_2$ such that 
$\mbox{supp} \, d \, h_1 \subset \Region_1 \bigcap \Region_2^\perp$,
where $\Region^\perp$ denotes the spacelike complement of 
$\Region$. We need to show that there is some function 
$\widetilde{h}_1 \in \Dc_1(\RR^4)$ with support 
in $\Region_1 \bigcap \Region_2^\perp$ such that $d(h_1 - \widetilde{h}_1) = 0$.
Since $d h_1$ is a closed two-form, this is a problem
in second cohomology with compact support 
for this region \cite{BoTu}, \ie we must show that 
the cohomology group $H_c^2(\Region_1 \bigcap \Region_2^\perp)$ is trivial. 
Now $\Region_1 \bigcap \Region_2^\perp$ is homeomorphic to the 
region $S^1 \times B^3$, where $B^3$ is a three-dimensional open ball,
and consequently 
$H_c^2(\Region_1 \bigcap \Region_2^\perp)  \approx H_c^2(S^1 \times B^3)$.
Next, by the K\"unneth formula for compact cohomology 
\cite[Sec.\ 5.19]{GrHaVa}, one has for integers $i,j \geq 0$
$$
H_c^2(S^1 \times B^3) \approx \bigoplus_{i + j = 2} 
H_c^i(S^1) \otimes H_c^j(B^3) \, .  
$$
Now $H_c^p(S^1) = \{ 0 \}$ for $p \geq 2$ since 
all p-forms of degree greater than $1$ vanish on 
$S^1$. Moreover, $H_c^p(B^3) \approx H_c^p(\RR^3)$ since $B^3$ is homeomorphic 
to $\RR^3$, and $H_c^p(\RR^3)= \{ 0 \}$ if $p \neq 3$ \cite[p.~19]{BoTu}.
Hence 
$H_c^2(\Region_1 \bigcap \Region_2^\perp) \approx H_c^2(S^1 \times B^3) 
= \{ 0 \}$. 
This implies that there is some test function 
$\widetilde{h}_1 \in \Dc_1(\RR^4)$ with support in 
$\Region_1 \bigcap \Region_2^\perp$ such that $d(h_1 - \widetilde{h}_1) = 0$ on 
$\RR^4$. Thus, according to the compact Poincar\'e lemma, there is some 
scalar test function $s \in \Dc_0(\RR^4)$ such that 
$h_1 = \widetilde{h}_1 + d s$. Because of current conservation 
one has $j(ds) = 0$, hence, taking into account 
the support properties of $\widetilde{h}_1, h_2$ and the
fact that $j$ is a local Wightman field, one finally gets
$[j(h_1), j(h_2)] = [j(\widetilde{h}_1), j(h_2)] = 0$.
This proves the statement in the special case. 

Now let $h_2$ be an arbitrary test function such that 
$\mbox{supp} \, h_2 \perp \mbox{supp} \, d \, h_1$. 
By a partition of unity, we decompose $h_2$ into a finite 
sum $h_2 = \sum_m h_{2 m}$ of test functions \mbox{$h_{2 m} \in \Dc_1(\RR^4)$}
having support in double cones $\Region_{2,m} \perp \mbox{supp} \, h_1$,
$\, m = 1, \dots, n$. Since $j$ is linear, we obtain 
$[j(h_1), j(h_2)] = \sum_m [j(h_1), j(h_{2 m})] = 0$ where  
the second equality follows from the preceding step,
completing the proof. 
\end{proof}

Turning now to the proof that 
$A_J$ satisfies condition \eqref{a3}, let 
$g_1, g_2 \in \Cc_1(\RR^4)$ with \  
$\mbox{supp} \, g_1 \perp \mbox{supp} \, g_2$.
Since $A_J(g_\frown) = 0$ by definition, it suffices to consider the cases 
where $g_1, g_2$ are both of type $\smile$. One then gets  
$[A_J(g_1), A_J(g_2)] = [j(h_1), j(h_2)]$ for their respective    
pre-images $h_1, h_2 \in \Dc_1(\RR^4)$. As was shown in the proof
of part~(ii) of Lemma \ref{l5.1}, one has \ 
$\mbox{supp} \, d h_m \subset \Region_m \doteq (\mbox{supp} \, g_m + V_+) \bigcap
 (\mbox{supp} \, g_m + V_-)$, \ $m = 1,2$, and consequently 
$\Region_1 \perp \Region_2$. We define now for given $y \in \RR^4$ functions 
$k_{m \, , y} \in \Dc_1(\RR^4)$, $m = 1,2$, putting
$$
x \mapsto k_{m \, , y}^\mu(x) \doteq 
\int_0^1 \! dt \, (d h_m)^{\mu \nu}(x - ty) \, y_\nu 
= \int_0^1 \! dt \, 
(\partial^\mu h_m^\nu(x - ty) - \partial^\nu h_m^\mu(x - ty)) \,
y_\nu \, .
$$
The integral can be computed, giving
$$ 
k_{m \, , y}^\mu(x) = h_m^\mu(x) - h_m^\mu(x - y) + \partial^\mu 
\int_0^1 \! dt \, y_\nu \, h_m^\nu(x - ty) \, ,
\quad m =1,2 \, .
$$ 
Note that the last term is the gradient of a scalar test function
which consequently gives zero if plugged into the current $j$ because of 
current conservation. From the integral representation of the 
functions $k_{m \, , y}$ one sees that these functions have support in 
the regions $\, \widetilde{\Region}_{m y} \doteq 
\bigcup_{\, 0 \leq t \leq 1} (\Region_m + ty)$, $\, m = 1,2$.
So, for sufficiently small $y \in \RR^4$, these regions are still
spacelike separated. We compute now for such $y$ 
\begin{align*}
& [j(h_1), j(h_2)] = [j(h_{1 y} + k_{1 \, , y}), j(h_{2 y} + k_{2 \, , y})] & \\
& = [j(h_{1 y}), j(h_{2 y})] 
+ [j(h_{1 y}), j(k_{2 \, , y})] 
+ [j(k_{1 \, , y}), j(h_{2 y})]    
+ [j(k_{1 \, , y}), j(k_{2 \, , y})] \, . &
\end{align*}
Since $\widetilde{\Region}_{1 y} \perp \widetilde{\Region}_{2 y}$, the last 
term on the right hand side of this equality vanishes 
because $j$ is local. Noticing that 
$\mbox{supp} \, d h_{1 y} \subset \widetilde{\Region}_{1 y}$
and $\mbox{supp} \, d h_{2 y} \subset \widetilde{\Region}_{2 y}$,
the first and second to last terms also vanish by the preceding 
lemma. So we are left with the equality 
$[j(h_1), j(h_2)] = [j(h_{1 y}), j(h_{2 y})]$
for small translations; iterating this procedure 
we then see that it holds for all $y \in \RR^4$. 
Because of locality, the commutator is thus 
an element of the centre of the algebra generated by the current.
Since the subspace $\delta d \Dc_1(\RR^4)$ is stable
under translations we conclude that this holds 
also for $[A_J(g_1) , A_J(g_2)]$ if 
$\mbox{supp} \, g_1 \perp \mbox{supp} \, g_2$,
establishing condition \eqref{a3}. 

It remains to show that the potential $A_J$ arises from
some regular vacuum representation of the universal algebra. 
To prove this let $g \in \Cc_1(\RR^4)$ be of
type $\smile$ and let 
$P \in \Poin_+^\uparrow$ be any Poincar\'e transformation. 
Since $(\delta d h)_P = \delta d h_P$ for $h \in \Dc_1(\RR^4)$
and since the connectivity properties of supports 
are respected by Poincar\'e transformations,
one obtains $g_{\smile P} = g_{P \smile}$. 
It follows that $A_J$ transforms covariantly under the 
unitary representation $U_J$ of the Poincar\'e group for the given current, 
\begin{align*}
U_J(P) A_J(g) U_J(P)^{-1} & = U_J(P) A_J(g_\smile) U_J(P)^{-1} =  
U_J(P) j(h_\smile) U_J(P)^{-1} & \\
& =  j(h_{\smile P}) = A_J(g_{\smile P}) = A_J(g_{P \smile}) = A_J(g_P) \, . &
\end{align*}
Thus the vector $\Omega_J \in \Hil_J$ represents a vacuum state 
also for the vector potential~$A_J$. We therefore obtain on the underlying
Hilbert space $\Hil_J$ a vacuum representation $\pi_J$ of the 
universal algebra $\fV$, putting 
$$
\pi_J(V(a,g)) \doteq e^{\, ia A_J(g_\smile)} = e^{\, ia j(h_{\smile})} \, ,
\quad a \in \RR, \ g \in \Cc_1(\RR^4) \, .
$$
The  generating functional of the vacuum state $\omega_J$ 
on $\fV$ is given by 
$$
\omega_J(V(a,g)) \doteq \langle \Omega_J, e^{i a A_J(g_\smile)} 
\Omega_J \rangle
=  \langle \Omega_J, e^{i a j(h_\smile)} \Omega_J \rangle \, , 
\quad a \in \RR, \ g \in \Cc_1(\RR^4) \, . 
$$
We summarise these results in the following proposition.

\vspace*{1mm} 
\begin{proposition}
Let $j$ be any local, covariant and conserved current 
satisfying all Wightman axioms and having domain properties  
as described above. There exists a regular vacuum 
state $\omega_J$ on the universal algebra $\fV$ of the electromagnetic
field such that the vector potential 
$A_J$ in the resulting GNS-representation is spacelike linear 
and satisfies the inhomogeneous Maxwell equation for the 
given current.
\end{proposition}

\vspace*{1mm} 
Let us finally determine the topological charge of
the potential $A_J$ associated with linked spacelike separated 
loop-shaped regions. So let $g \in \Cc_1(\RR^4)$ have support
in some loop-shaped region $\Loop$ which can continuously be retracted to 
some simple spacelike loop $\gamma$. This is then also possible 
for 
\mbox{$\widetilde{\Loop} \doteq \{ \Loop + V_+ \} \bigcap \{ \Loop + V_- \} 
\supset \Loop$}. So the complement 
$\RR^4 \backslash \widetilde{\Loop}$ is homotopic to 
$\RR^4 \backslash \gamma$ and one has for the corresponding 
homology groups
$ H_1(\RR^4 \backslash \widetilde{\Loop}) \approx H_1(\RR^4 \backslash \gamma)$.
By Alexander duality one obtains 
$H_1(\RR^4 \backslash \gamma) \approx H^2(\gamma) \approx H^2(S^1) 
= \{ 0 \}$, cf. \cite[Ch.~VI, Cor.~8.6]{Br}. 
So we conclude that for 
the $\smile$ component $g_\smile$ of $g$, which likewise has support
in~$\widetilde{\Loop}$, there exists according to part (ii) of
Lemma \ref{l5.1} a pre-image $h_\smile$ which has support in any 
given neighbourhood of $\widetilde{\Loop}$. Thus if 
$g_1, g_2 \in \Cc_1(\RR^4)$ have their supports in spacelike
separated loop-shaped regions $\Loop_1$ and $\Loop_2$, respectively,
one obtains 
$$
\, [A_J(g_1), A_J(g_2)] = [A_J(g_{1 \smile}), A_J(g_{2 \smile})]
= [j(h_{1 \smile}), j(h_{2 \smile})] = 0 \, .
$$
The latter equality 
follows from the fact that $h_1$, $h_2$ have their supports in 
neighbourhoods of the spacelike separated regions
$\widetilde{\Loop}_1$ and $\widetilde{\Loop}_2$, respectively, 
and the locality of the current. Thus the topological charges
related to linked loop-shaped regions turn out to be trivial for the 
vector potentials $A_J$. 

One can, however, combine the results of the preceding section 
with the present ones in order to exhibit examples of 
regular vacuum representations of the universal algebra $\fV$
with, both, non-trivial topological charges and electric currents.
The construction is similar to the definition of s-products in the
Wightman framework of quantum field theory \cite{Bo}. For the 
convenience of the reader, we briefly recall here this procedure.
Let $(\pi_T, \Hil_T, \Omega_T)$ be 
a regular vacuum representation of $\fV$
with non-trivial topological charge 
but trivial electric current and
let $(\pi_J, \Hil_J, \Omega_J)$ be a regular vacuum representation 
of $\fV$ with non-trivial electric current
but trivial topological charge. One then constructs a representation 
$\pi_{T J}$ on $\Hil_T \otimes \Hil_J$, putting 
$\pi_{T J}(V(a,g)) \doteq \pi_{T}(V(a,g)) \otimes \pi_{J}(V(a,g))$
for $a \in \RR$, $g \in \Cc_1(\RR^4)$.
Restricting the algebra $\pi_{T J}(\fV)$ to the subspace 
$\Hil_{T J} \doteq \pi_{T J}(\fV) \, \Omega_{T J} \subset \Hil_T \otimes \Hil_J$,
where $\Omega_{T J} \doteq \Omega_T \otimes \Omega_J$, 
one obtains a 
regular vacuum representation $(\pi_{T J}, \Hil_{T J}, \Omega_{T J})$ 
of $\fV$ with generating function 
$$
\omega_{T J}(V(a,g))
\doteq \langle \Omega_T, \pi_T(V(a,g)) \Omega_T \rangle \,
\langle \Omega_J, \pi_J(V(a,g)) \Omega_J \rangle \, ,
\quad a \in \RR,  \ g \in \Cc(\RR^4) \, .
$$
We omit the simple proof that this representation has the
desired properties.

\section{Topological charges of multiplets of electromagnetic 
fields}

\setcounter{equation}{0}

Up to this point we have considered the case of 
a single species of the electromagnetic field, described by the universal 
algebra $\fV$. One may expect that 
multiplets of such fields appear in the short distance (scaling)
limit of non-abelian gauge theories, where the gauge fields become 
non-interacting and transform as tensors under the adjoint 
action of some global gauge group. It seems therefore worthwhile 
to discuss the possible appearance of topological charges for these limit 
theories.   

In order to keep the framework simple, we treat here only the case of 
two electromagnetic fields with corresponding  intrinsic vector potentials.
The description of these potentials 
in terms of a universal algebra is straight forward: 
instead of proceeding from the space $\Cc_1(\RR^4)$ of test functions,
we consider now the algebraic sum 
$\Cc_1(\RR^4) \oplus \Cc_1(\RR^4)$. The corresponding universal 
algebra is determined by the relations \eqref{a1}, \eqref{a2} and \eqref{a3}, 
where one now takes \mbox{$g \in \Cc_1(\RR^4) \oplus \Cc_1(\RR^4)$} 
and where the support properties are determined by the eight 
components of these functions. It follows from the arguments given 
in \cite{BuCiRuVa} that the *-algebra generated by the unitaries
$V_2(a,g)$, where $a \in \RR$, $g \in \Cc_1(\RR^4) \oplus \Cc_1(\RR^4)$,
has a C*-norm induced by all of its GNS representations. 
The completion of this algebra with regard to this 
norm will be denoted by $\fV_2$; note that this algebra  
is not isomorphic to the completion of the algebraic tensor 
product $\fV \otimes \fV$ since 
the two types of fields need not commute with each other. 
The definitions of Poincar\'e 
transformations and of regular vacuum states 
carry over, \textit{mutatis mutandis}, to $\fV_2$ from 
those given for $\fV$. Thus they need not be detailed here. 

It is our aim to show that there exist regular vacuum representations
of $\fV_2$ with non-trivial topological charge.
More concretely, the group-theoretic commutator 
$\lfloor V_2(g_1), V_2(g_2) \rfloor$ attains 
in these representations values different from 
$1$ for test functions $g_1, g_2 \in \Cc_1(\RR^4) \oplus \Cc_1(\RR^4)$ 
with $\mbox{supp} \, g_1 \perp \mbox{supp} \, g_2$. Yet, in 
contrast to the results of Sec.\ \ref{no-go} for the algebra $\fV$, 
the underlying electromagnetic fields are now (linear)
Wightman fields. 

The example which we are going to construct is 
a generalised free field theory 
which is determined by its two-point function. In order
to simplify the notation, we introduce subspaces of 
$\Cc_1(\RR^4) \oplus \Cc_1(\RR^4)$, given by 
$\Cc_u(\RR^4) = \Cc_1(\RR^4) \oplus \{ 0 \}$, 
$\Cc_d(\RR^4) = \{ 0 \} \oplus \Cc_1(\RR^4)$ and we 
denote by $g_{u}$, $g_{d}$ their respective elements. 
Moreover,  similarly to Sec.\ \ref{go-go},  
we proceed from the components \ $g \in \Cc_1(\RR^4)$ 
to their (classes of) co-primitives $G \in \Dc_2(\RR^4)$, 
\viz $g = \delta G$. Making use of this notation, we define
for fixed $-1 \leq \zeta \leq 1$ a sesquilinear form on the
complex vector space \
$\CC \, \Cc_1(\RR^4) \oplus \Cc_1(\RR^4)$. Given 
$g_m = (g_{m \, u} + g_{m \, d})$ with co-primitive 
$G_m = (G_{m \, u} + G_{m \, d})$, $m = 1,2$, we put
$$
\langle g_1, g_2 \rangle_\zeta \doteq
\langle G_{1 u}, G_{2 u} \rangle_0
+ \langle G_{1 d}, G_{2 d} \rangle_0 
+ \zeta \langle G_{1 u}, \star G_{2 d} \rangle_0 
- \zeta \langle G_{1 d}, \star G_{2 u} \rangle_0 \, .
$$
Here 
$$
\, \langle G_1 , G_2 \rangle_0 \doteq (2 \pi)^{-3} \int \! dp \, 
\theta(p_0) \, \delta(p^2) \, 
\overline{(p \, \widehat{G_1} (p))} (p \, \widehat{G_2}(p))
$$
is the familiar scalar product on the one-particle space of the
free Maxwell field for $G_1, G_2 \in \CC \, \Dc_2(\RR^4)$.  Since 
$\, \langle \star G_1 , \star G_2 \rangle_0 = 
\langle G_1 , G_2 \rangle_0 $ and
$\star \star \upharpoonright  \CC \, \Dc_2(\RR^4) = -1$ 
one has 
$\overline{\langle g_1, g_2 \rangle}_\zeta = \langle g_2, g_1 \rangle_\zeta$.
Moreover, 
$$
| \langle G_1, \star G_2 \rangle_0 |^2 
\leq \langle G_1, G_1 \rangle_0 \, \langle \star G_2 , \star G_2 \rangle_0 
= \langle G_1, G_1 \rangle_0 \, \langle G_2 , G_2 \rangle_0 
$$ 
and consequently
$$
\langle g, g \rangle_\zeta \geq 
\langle G_{u}, G_{u} \rangle_0 + \langle G_{d}, G_{d} \rangle_0
- 2 | \zeta | \sqrt{\langle G_{u}, G_{u} \rangle_0 \langle G_{d}, G_{d} \rangle_0}
\, \geq 0 \, .
$$
Hence $\langle \, \cdot \, , \, \cdot \, \rangle_\zeta$
defines a positive (semidefinite) scalar product on \
$\CC \, \Cc_1(\RR^4) \oplus \Cc_1(\RR^4)$ for the given range
of $\zeta$. Since it is invariant under Poincar\'e transformations
and the Fourier transforms of 
$x \mapsto \langle g_1 , g_{2 \, x} \rangle_\zeta$
have support in $V_+$, it can be taken as  
two-point function of (a pair of) vector potentials $A_\zeta$.
Moreover, one can define corresponding regular quasi-free vacuum states 
$\omega_{ \, \zeta}$ on $\fV_2$ with generating function 
$$
\omega_{ \, \zeta}(V_2(a,g)) \doteq e^{- a^2 \, \langle g, 
g \rangle_\zeta \, / 2 } \, ,
\quad a \in \RR \, , \ g \in \Cc_1(\RR^4) \oplus \Cc_1(\RR^4) \, .
$$
Proceeding to the GNS representations $(\pi_\zeta, \Hil_\zeta, \Omega_\zeta)$,
one obtains vector potentials $A_\zeta(g)$ which depend linearly on 
$g \in \Cc_1(\RR^4) \oplus \Cc_1(\RR^4)$.  
Furthermore, the resulting electromagnetic fields satisfy all 
Wightman axioms \cite{StWi}. The familiar arguments establishing
these facts are omitted here. 

Let us turn now to the determination of the topological charges carried
by $A_\zeta$. They are fixed by the commutator, 
$g_1, g_2 \in  \Cc_1(\RR^4) \oplus \Cc_1(\RR^4)$, 
\begin{align*}
& [A_\zeta(g_1), A_\zeta(g_2)] = (\langle g_1, g_2 \rangle_\zeta  -
\langle g_2, g_1 \rangle_\zeta ) \, 1 & \\
& = \big( \Delta(G_{1 u} , G_{2 u}) +  \Delta(G_{1 d} , G_{2 d})
+ \zeta \Delta(G_{1 u} , \star G_{2 d}) 
- \zeta \Delta(G_{1 d} , \star G_{2 u}) \big) \, 1  
\, . &
\end{align*}
Here $\Delta$ is the distribution given in relation 
\eqref{pauli-jordan}. Whence these commutators have a similar structure
as the ones given in 
equation \eqref{commutator} for the topological vector potential $A_T$. 
In  particular,  if $\zeta \neq 0$, they contain 
terms of the form $\Delta(G, \star G^\prime )$ which are different from  
zero for suitable functions 
$ g = \delta \, G, g^\prime = \delta \, G^\prime \in \Cc_1(\RR^4)$ with
$\mbox{supp} \, g \perp \mbox{supp} \, g^\prime $, \ cf.~Sec.~\ref{go-go}. 
Thus, choosing in the present case a pair $g_1 \in \Cc_u(\RR^4)$,
$g_2 \in \Cc_d(\RR^4)$ of this particular type, it follows that 
the commutators are different from zero, indicating
a non-trivial topological charge. Hence the present model
provides examples of such theories with fully linear 
vector potentials. We note that these theories have, for any 
$-1 \leq \zeta \leq 1$, a global internal
symmetry group SO(2) whose action \ $R(\theta)$, $\theta \in [0, 2 \pi]$, 
on \ $g \in \Cc_1(\RR^4) \oplus \Cc_1(\RR^4)$ is given by 
$$
(R(\theta)g)_u = \cos(\theta) g_u + \sin(\theta) g_d \, , \quad
(R(\theta)g)_d = -sin(\theta) g_u + \cos(\theta) g_d \, , 
$$
in an obvious notation.

In a similar manner one can exhibit the existence of representations 
with non-trivial topological charge for higher dimensional 
multiplets of electromagnetic 
Wightman fields. The details are straightforward 
generalisations of the present example.
Of particular interest is the question in which manner 
such representations could 
manifest themselves already at finite scales in asymptotically free
non-abelian gauge theories. To answer 
this question would go beyond the 
scope of the present investigation, however. 

\section{Conclusions}
\setcounter{equation}{0}

Within the framework of the universal C*-algebra of the electromagnetic 
field, we have clarified the conditions for the existence of
theories with topological charges.  
As we have seen, the intrinsic (gauge invariant) 
vector potential cannot depend
linearly on test functions if these charges are non-trivial. On the other
hand, there exist many examples of regular vacuum representations of 
the algebra, carrying topological charges, where the vector potential still 
satisfies a condition of spacelike linearity. Thus,  there are no 
fundamental obstructions to the appearance of such charges.
 
It is an intriguing question how the existence of 
these charges would manifest itself experimentally. 
Roughly speaking, such non-trivial charges would indicate
that certain linked loop-shaped configurations of the 
(intrinsic) vector potential, 
which at first sight ought to be commensurable because of their spacelike 
separation, do not comply with this condition.  
One then has non-trivial commutation relations of canonical 
type for these potentials. Hence one might expect
in experiments that coherent photons, traversing a loop in 
the complement of another electromagnetic loop, would exhibit 
interference patterns, akin to the Aharonov-Bohm effect.

It is known that stable loop-shaped configurations of the (classical)
electromagnetic field exist \cite{BoIr}. Thus, in principle,  
such experiments seem feasible. In fact, topologically non-trivial 
configurations of quantum matter have recently been observed 
\cite{GhHaMoRaRuTi}. If one could extend these experiments to the 
quantised electromagnetic field, the question of the existence of 
non-trivial topological charges would become amenable
to an answer. 

We conclude this article by addressing a seeming theoretical puzzle
related to our results. 
As has been shown in \cite{FrHe,Bost}, 
one can construct from any Haag-Kastler theory 
under quite general conditions certain Wightman fields.
This is accomplished by proceeding from the spatially 
extended observables to point-like limits, integrating 
them along the way with test functions. 
As a matter of fact, the present examples nicely illustrate this fact.  
Applying these methods to the algebra generated 
by the topological vector 
potential~$A_T$, constructed in Sec.~\ref{go-go}, one would recover the
free Maxwell field; similarly, one would recover from the algebra
generated by the potential $A_J$, constructed in Sec.~\ref{current},   
the underlying current. Thus it may seem that the usage 
of non-linear fields is not really necessary. 

There are, however, two reasons why such a conclusion would be premature. 
First, in the framework of the universal algebra, the 
interpretation of the generating unitaries 
as exponentials of the intrinsic vector potential has been fixed
from the outset. These operators ought to be identified with corresponding 
hard ware measuring this field. Thus the transition from 
these operators to related 
Wightman fields would in general amount to a quite different physical 
interpretation of the theory. Second, point-like Wightman fields 
frequently hide the topological aspects of the theory which they 
generate. As a matter of fact, it took almost 50 years 
after the invention of the quantised electromagnetic field 
until Roberts exhibited in \cite{Roberts} certain of its topological
features. 
In contrast, these topological 
properties are already encoded in the defining relations 
of the universal algebra,  
in particular in relation \eqref{a3}. These relations   
were established in \cite{BuCiRuVa} by general, model-independent
considerations. We therefore think that the universal algebra of the
electromagnetic field and its possible generalisation to multiplets of
such fields provides an adequate framework for further study of
topological aspects of gauge theories.   

\vspace*{-1mm}
\section*{Acknowledgement}

\vspace*{-2mm}
DB  gratefully acknowledges the hospitality and financial support 
extended to him by the University of Rome ``Tor Vergata'' which made
this collaboration possible. 
FC and GR are supported by the ERC Advanced Grant 
669240 QUEST ``Quantum Algebraic Structures and Models". 
EV is supported in part by OPAL ``Consolidate the Foundations''.


\begin{thebibliography}{22}
{\small 

\bibitem{ArHaSch} Araki, H., Haag, R.\ and Schroer, B.,
``The determination of a local or almost local field from a 
given current'',
Nuovo Cimento \textbf{19} (1961) 90-102

\bibitem{Bo}
Borchers, H.J., ``Algebras of unbounded operators in quantum field theory'',
Physica A \textbf{124} (1984) 127-144

\bibitem{Bost}
Bostelmann, H., ``Phase space properties and the short distance structure in 
quantum field theory'', J. Math. Phys. \textbf{46} (2005) 052301 

\bibitem{Br}
Bredon, G.B.,
{\it Topology and Geometry}, Springer, New York, 1993

\bibitem{BrDuFr}
Brunetti, R., D\"utsch, M.\ and Fredenhagen, K., 
``Algebraic quantum field theory and the renormalization 
groups'', 
Advances in Theoretical and
Mathematical Physics \textbf{13} (2009) 1541-1599

\bibitem{BoTu} Bott, R.\ and Tu, L.W.,
{\it Differential Forms in Algebraic Topology}, 
Graduate Texts in Mathematics \textbf{82}, 
Springer, New York, 1982

\bibitem{BuCiRuVa} Buchholz, D., Ciolli, F., Ruzzi, G. and Vasselli, E., 
``The universal C*-algebra of the electromagnetic field'',
Lett.\ Math.\ Phys.\ {\bf 106} (2016) 269--285 \ 
Erratum: Lett.\ Math.\ Phys.\ {\bf 106} (2016) 287

\bibitem{BuMaPaTo} Buchholz, D., Mack, G., Paunov, R.R. and Todorov, I.T.,
``An algebraic approach to the classification of local conformal field 
theories'', pp. 299-305 in: 
IXth International Congress on Mathematical Physics. Swansea 1988,
Eds. I.M. Davies, B. Simon, A.~Truman, Adam Hilger, Bristol 1989

\bibitem{BoIr} Bouwmeester, D.\ and Irvine, W.T.M.,
``Linked and knotted beams of light'', 
Nature Physics {\bf 4} (2008) 716-720 

\bibitem{FrHe} Fredenhagen, K.\ and Hertel, J., 
``Local algebras of observables and pointlike localized fields'',
Commun.\ Math.\ Phys.\ {\bf 80} (1981) 555-561

\bibitem{GhHaMoRaRuTi} Gheorghe, A.H., Hall, D.S., M\"ott\"onen, Ray, M.W., 
Ruokowski, E.\ and Tiurev, K.,
``Tying quantum knots'', 
Nature Physics {\bf 12} (2016) 478-483

\bibitem{GrHaVa} Greub, W., Halperin, S.\ and 
Vanstone, R., 
{\it Connections, curvature, and cohomology.~V.1}, 
Academic Press 1972

\bibitem{HaKa} Haag, R.\ and Kastler, D., 
``An algebraic approach to quantum field theory'', 
J.\ Math.\ Phys.\ {\bf 5} (1964) 848--861

\bibitem{Hat} Hatcher, A.,  
{\it Algebraic Topology}, 
Cambridge University Press, Cambridge England, 2002 

\bibitem{Roberts} Roberts, J.E., 
``A survey of local cohomology'', \ 
In: Mathematical problems in theoretical physics. (Rome, 1977), 81--93, 
Lecture Notes in Phys. 80, Springer, Berlin, New York, 1978

\bibitem{StWi}  Streater, R.F.\ and Wightman, A.S., 
{\it PCT, Spin and Statistics, and All That},
W.A. Benjamin, New York, Amsterdam, 1964

\bibitem{Steinmann} Steinmann, O., 
{\it Perturbative Quantum Electrodynamics and Axiomatic Field Theory},
Springer, Berlin, Heidelberg, 2000

\bibitem{Strocchi2} Strocchi, F.,
{\it An Introduction to Non--Perturbative Foundations of Quantum 
Field Theory},
Oxford University Press, 2013

}
\end{thebibliography}
\end{document}